\documentclass[a4paper,USenglish,cleveref, autoref,thm-restate]{lipics-v2021}

\nolinenumbers
\hideLIPIcs
\usepackage[utf8]{inputenc}
\usepackage{tikz}
\tikzstyle{vertex}=[draw, circle, fill, inner sep = 2.4pt]

\usepackage{todonotes}

\newtheorem{rrule}[theorem]{Reduction Rule}
\crefname{rrule}{Reduction Rule}{Reduction Rules}
\newtheorem{brule}[theorem]{Branching Rule}
\crefname{brule}{Branching Rule}{Branching Rules}
\newtheorem{term}[theorem]{Termination}
\crefname{term}{Termination}{Termination}

\newcommand{\VC}{\ensuremath{\mathsf{VC}}}
\newcommand{\IS}{\ensuremath{\mathsf{IS}}}
\newcommand{\MM}{\ensuremath{\mathsf{MM}}}
\newcommand{\IM}{\ensuremath{\mathsf{IM}}}
\newcommand{\UB}{\ensuremath{\mathsf{UB}}}
\newcommand{\LP}{\ensuremath{\mathsf{LP}}}
\newcommand{\Ostar}{\ensuremath{O^{\star}}}

\newcommand{\problem}[4]{
    \vspace{.8ex}
    {\centering
        \fbox{~\begin{minipage}{.97\textwidth}
            \vspace{2pt} 
            \noindent
            \normalsize\textsc{#1}
            \vspace{1.5pt}

            \setlength{\tabcolsep}{1pt}
            \renewcommand{\arraystretch}{1.0}
            \begin{tabularx}{\textwidth}{@{}lX@{}}
                \normalsize\textbf{Input:}       & \normalsize#2 \\
                \normalsize\textbf{Question:}    & \normalsize#3 \\
                \normalsize\textbf{Parameter: }  & \normalsize#4
            \end{tabularx}
        \end{minipage}}
    }
    \vspace{1.2ex}
}

\title{Induced Matching below Guarantees: Average Paves the Way for Fixed-Parameter Tractability}
\titlerunning{Induced Matching below Guarantees}

\author{Tomohiro Koana}{Technische Universität Berlin, Faculty~IV, Institute of Software Engineering and Theoretical Computer Science, Algorithmics and Computational Complexity}{tomohiro.koana@tu-berlin.de}{0000-0002-8684-0611}{Supported by the DFG Project DiPa, NI 369/21.}

\authorrunning{T. Koana}


\acknowledgements{The author would like to thank Niclas Boehmer for his feedback that improved the presentation of this paper.}

\Copyright{Tomohiro Koana}


\begin{document}

\maketitle

\begin{abstract}
In this work, we study the \textsc{Induced Matching} problem:
Given an undirected graph $G$ and an integer $\ell$, is there an induced matching $M$ of size at least $\ell$?
An edge subset $M$ is an induced matching in $G$ if $M$ is a matching such that there is no edge between two distinct edges of $M$.
Our work looks into the parameterized complexity of \textsc{Induced Matching} with respect to ``below guarantee'' parameterizations.
We consider the parameterization $u - \ell$ for an upper bound $u$ on the size of any induced matching.
For instance, any induced matching is of size at most $n / 2$ where $n$ is the number of vertices, which gives us a parameter $n / 2 - \ell$.
In fact, there is a straightforward $9^{n/2 - \ell} \cdot n^{O(1)}$-time algorithm for \textsc{Induced Matching} [Moser and Thilikos, J. Discrete Algorithms].
Motivated by this, we ask:
Is \textsc{Induced Matching} FPT for a parameter smaller than $n / 2 - \ell$?
In search for such parameters, we consider $\MM(G) - \ell$ and $\IS(G) - \ell$, where $\MM(G)$ is the maximum matching size and $\IS(G)$ is the maximum independent set size of $G$.
We find that \textsc{Induced Matching} is presumably not FPT when parameterized by $\MM(G) - \ell$ or $\IS(G) - \ell$.
In contrast to these intractability results, we find that taking the average of the two helps---our main result is a branching algorithm that solves \textsc{Induced Matching} in $49^{(\MM(G) + \IS(G))/ 2 - \ell} \cdot n^{O(1)}$ time.
Our algorithm makes use of the Gallai--Edmonds decomposition to find a structure to branch on.
\end{abstract}

\section{Introduction}

A matching in a graph is a set of pairwise non-incident edges.
An induced matching is a matching such that no edge is incident with two edges from the matching.
The notion of induced matchings was initially introduced by Stockmeyer and Vazirani \cite{DBLP:journals/ipl/StockmeyerV82}.
Since then, the \textsc{Induced Matching} problem---given an undirected graph $G$ and an integer $\ell \in \mathbb{N}$, we are to determine whether $G$ has an induced matching of size $\ell$---has been studied extensively.
This problem is NP-hard, which was proven independently by Stockmeyer and Vazirani \cite{DBLP:journals/ipl/StockmeyerV82} and Cameron~\cite{DBLP:journals/dam/Cameron89}.
The NP-hardness persists on restricted graph classes, such as bipartite graphs of vertex degree at most three~\cite{DBLP:journals/ipl/Lozin02} and cubic planar graphs~\cite{DBLP:journals/jda/DuckworthMZ05}.
On the positive side, \textsc{Induced Matching} is polynomial-time solvable on trees \cite{DBLP:conf/wg/Zito99}, chordal graphs \cite{DBLP:journals/dam/Cameron89}, weakly chordal graphs~\cite{DBLP:journals/dm/CameronST03}, circular-arc graphs \cite{DBLP:journals/dam/GolumbicL93}, comparability graphs \cite{DBLP:journals/dam/GolumbicL00}, and AT-free graphs~\cite{DBLP:journals/dm/Cameron04}.

In this work, we study the parameterized complexity of \textsc{Induced Matching}.
The standard parameterization of \textsc{Induced Matching} takes the solution size $\ell$ as the parameter.
For $\ell$, \textsc{Induced Matching} is W[1]-hard, which can be easily seen by a parameterized reduction from the W[1]-hard \textsc{Independent Set} problem.
Indeed, the W[1]-hardness holds even on bipartite graphs~\cite{DBLP:journals/dam/MoserS09}.
Perhaps for this reason, many researchers have investigated the parameterized complexity of \textsc{Induced Matching} from multivariate perspectives.
Moser and Sikdar~\cite{DBLP:journals/dam/MoserS09} gave a kernel of size $O(\Delta^3 \ell)$ for \textsc{Induced Matching} for the maximum degree $\Delta$.
Erman et al.~\cite{DBLP:journals/dam/ErmanKKW10} and Kanj et al.~\cite{DBLP:journals/jcss/KanjPSX11} independently found that \textsc{Induced Matching} admits a polynomial kernel of size $\ell^{O(d)}$ on $d$-degenerate graphs.
This result was later complemented by Cygan et al.~\cite{DBLP:journals/talg/CyganGH17}, who showed that the kernel size is basically tight---there is no kernel of size $\ell^{o(d)}$ under standard complexity assumptions.
Recently, it was shown that \textsc{Induced Matching} on $c$-closed graphs~\cite{DBLP:conf/esa/KoanaKS20} and on weakly $\gamma$-closed graphs~\cite{DBLP:conf/isaac/KoanaKS21} has a kernel with at most $O(c^7 \ell^8)$ vertices and $\ell^{O(\gamma)}$ vertices, respectively.
We remark that the parameterized complexity on planar graphs has also received considerable attention~\cite{DBLP:journals/jda/MoserT09,DBLP:journals/dam/ErmanKKW10,DBLP:journals/jcss/KanjPSX11}.

In this paper, we adopt an alternative approach to tackle the fixed-parameter tractability of \textsc{Induced Matching}---using \emph{above guarantees} and \emph{below guarantees} \cite{DBLP:journals/jal/MahajanR99,DBLP:journals/jcss/MahajanRS09} (see also a very recent survey \cite{DBLP:journals/corr/abs-2207-12278}).
This work particularly concerns below guarantees.
For a graph $G$, let $\IM(G)$ denote the maximum induced matching size of $G$.
In a nutshell, we employ $\UB(G) - \ell$ as a parameter where $\UB(G)$ is a function on $G$ that upper-bounds $\IM(G)$, i.e., $\UB(G) \ge \IM(G)$ on every graph $G$.

In this spirit, we first consider a trivial upper bound.
For the number $n$ of vertices, $\frac{1}{2} n$ clearly constitutes an upper bound on $\IM(G)$.
We refer to the parameterized problem arising from this upper bound $\UB(G) = \frac{1}{2} n$ as \textsc{Induced Matching Below Trivial Guarantee}:

\problem
{Induced Matching Below Trivial Guarantee (IMBTG)}
{An undirected graph $G$ and an integer $\ell$.}
{Does $G$ have an induced matching of size $\ell$?}
{$k = \frac{1}{2} n - \ell$}

This problem has been studied in the literature, albeit under different names:
Moser and Thilikos \cite{DBLP:journals/jda/MoserT09} gave an algorithm solving \textsc{IMBTG} in $O^{\star}(9^{k})$ time.\footnote{The $O^{\star}$ notation suppresses the polynomial factor in the input size.}
Subsequently, Xiao and Kou \cite{DBLP:journals/tcs/XiaoK20} developed an algorithm running in $O^{\star}(3.1845^k)$ time.
In terms of kernelization, a kernel with $O(k^2)$ vertices was given by Moser and Thilikos \cite{DBLP:journals/jda/MoserT09}.
Later, Xiao and Kou \cite{DBLP:journals/tcs/XiaoK20} gave an improved kernel with $O(k)$ vertices.

In parameterized complexity, whenever the fixed-parameter tractability with respect to a parameter $k$ is discovered, one asks whether the ``boundary of tractability'' can be taken further, that is,  fixed-parameter tractability is achievable for a parameter $k'$ smaller than $k$ (i.e., $k' \le g(k)$ for some function $g$). (An analogous question arises when W-hardness for $k$ is discovered as well---does fixed-parameter parameter tractability hold for a parameter $k'$ larger than $k$?)
This question appears prominently in  multivariate algorithmics \cite{DBLP:conf/mfcs/KomusiewiczN12,DBLP:conf/stacs/Niedermeier10} and structural parameterizations~\cite{DBLP:conf/esa/Belmonte0LMO20,DBLP:journals/ejc/FellowsJR13}.
We ask ourselves this type of question for \textsc{Induced Matching} parameterized by below guaranteed values.
More precisely, the main question we challenge in this work is the following:
\vspace{-1ex}
\begin{quote}
($\star$) Is \textsc{Induced Matching} FPT for a parameterization smaller than that of \textsc{IMBTG}?
\end{quote}
\vspace{-1ex}

Let us remark we are not the first to address this kind of question in the context of above and below guarantee parameterizations:
The \textsc{Vertex Cover} problem---given an undirected graph $G$ and an integer $\ell$, decide whether there is a set of at most $\ell$ vertices that is incident with every edge---is one of the problems where this kind of question was considered.
A simple branching algorithm solves \textsc{Vertex Cover} in $\Ostar(2^\ell)$ time.
For any graph $G$, it holds that $\VC(G) \ge \LP(G) \ge \MM(G)$, where $\VC(G)$, $\MM(G)$, and $\LP(G)$ denote the minimum vertex cover size, the maximum matching size, and the optimum of the linear programming relaxation of \textsc{Vertex Cover}, respectively.
These inequalities give rise to the above-guarantee parameterizations, $\ell - \MM(G)$ and $\ell - \LP(G)$.
These parameterizations have been extensively studied \cite{DBLP:journals/toct/CyganPPW13,DBLP:conf/soda/IwataOY14,DBLP:journals/jacm/KratschW20,DBLP:journals/talg/LokshtanovNRRS14,DBLP:conf/stacs/NarayanaswamyRRS12,DBLP:conf/esa/RamanRS11,DBLP:journals/jcss/RazgonO09}---
it has been shown that \textsc{Vertex Cover} FPT is with respect to $\ell - \MM(G)$ as well as $\ell - \LP(G)$.
Notably, the border of tractability was further extended to $\ell - (2 \LP(G) - \MM(G))$ \cite{DBLP:conf/soda/GargP16,DBLP:journals/siamdm/Kratsch18}.
Let us also remark that the above guarantee parameterizations of \textsc{Max Cut} and related problems have been also extensively studied \cite{DBLP:journals/tcs/CrowstonGJM13,DBLP:journals/algorithmica/CrowstonJM15,DBLP:journals/algorithmica/EtscheidM18,DBLP:journals/mst/MadathilSZ20,DBLP:journals/jcss/MnichPSS14,DBLP:journals/jal/MahajanR99,DBLP:conf/wg/MnichZ12}.

Now we come back to the question ($\star$) on \textsc{Induced Matching}.
We wish to find an upper bound on the maximum induced matching size which is stricter than $\frac{1}{2} n$.
The maximum matching size $\MM(G)$ serves as such a bound:
Clearly, $\MM(G) \le \frac{1}{2} n$.
Moreover, $\MM(G)$ is an upper bound on $\IM(G)$, since every induced matching is a matching.
This leads to the following parameterized problem:

\problem
{Induced Matching Below Maximum Matching (IMBMM)}
{An undirected graph $G$ and an integer $\ell$.}
{Does $G$ have an induced matching of size $\ell$?}
{$k = \MM(G) - \ell$}

As graphs $G$ with $\MM(G) = \IM(G)$ have been of theoretical interest, there are known results on \textsc{IMBMM}:
Kobler and Rotics~\cite{DBLP:journals/algorithmica/KoblerR03} gave a polynomial-time algorithm for $k = 0$.
Cameron and Walker \cite{DBLP:journals/dm/CameronW05a} extended this result by providing a structural characterization of graphs $G$ with $\MM(G) = \IM(G)$.
Later, Duarte et al.~\cite{DBLP:journals/tcs/DuarteJPRS15} developed an algorithm for \textsc{IMBMM} that runs in $n^{O(k)}$ time.
However, it has been left open whether \textsc{IMBMM} is FPT.
Filling the gap, we prove in \Cref{sec:hard} that \textsc{IMBMM} is W[2]-hard, i.e., presumably not FPT.
This implies that taking $\MM(G)$ as the upper bound falls short to answer the question ($\star$).

Next, we consider another natural upper bound: the maximum independent set size $\IS(G)$.
Since an induced matching of size $\ell$ contains an independent set of size $\ell$, any graph satisfies $\IS(G) \ge \IM(G)$, which gives the following parameterization:

\problem
{Induced Matching Below Independent Set (IMBIS)}
{An undirected graph $G$ and an integer $\ell$.}
{Does $G$ have an induced matching of size $\ell$?}
{$k = \IS(G) - \ell$}

Although $\frac{1}{2} n$ and $\IS(G)$ are incomparable in general (consider complete graphs and empty graphs), the parameter of \textsc{IMBIS} is essentially smaller compared to that of \textsc{IMBTG}:
Observe that for $\ell > 0$, the graph $G'$ on $2n$ vertices obtained from $G$ by adding $n$ vertices adjacent to all other vertices fulfills $\IM(G') \ge \ell$ if and only if $\IM(G) \ge \ell$.
The maximum independent set size of $G'$ is at most $n$---half the number of vertices in $G'$.
Thus, if \textsc{IMBIS} was fixed-parameter tractability, then fixed-parameter tractability of \textsc{IMBTG} would follow, answering our question ($\star$).
We find, however, that \textsc{IMBIS} is NP-hard for $k = 0$ even if an independent set of size $\ell$ is provided as part of the input.

Somewhat dismayed by the previous two negative results, we look into the upper bound obtained by taking the average of $\MM(G)$ and $\IS(G)$, that is, $\frac{1}{2}(\MM(G) + \IS(G))$.
For any graph $G$, we have
\begin{align*}
  \frac{1}{2} (\MM(G) + \IS(G)) - \IM(G) = \frac{1}{2}(\MM(G) - \IM(G)) + \frac{1}{2} (\IS(G) - \IM(G)) \ge 0,
\end{align*}
implying that $\IM(G) \le \frac{1}{2}(\MM(G) + \IS(G))$.
(This equation also implies that the parameterization by $\frac{1}{2}(\MM(G) + \IS(G)) - \ell$ is larger compared to the parameterization of \textsc{IMBMM} and \textsc{IMBIS}.)
On the other hand, we have $\frac{1}{2}(\MM(G) + \IS(G)) \le \frac{1}{2}(\VC(G) + \IS(G)) \le \frac{1}{2} n$.
Hence, the parameterization obtained from the average is indeed smaller than the trivial below guarantee.
Formally, we study the following:

\problem
{Induced Matching Below Average (IMBA)}
{An undirected graph $G$ and an integer $\ell$.}
{Does $G$ have an induced matching of size $\ell$?}
{$k = \frac{1}{2}(\MM(G) + \IS(G)) - \ell$}

The main result of this work is an FPT algorithm for \textsc{IMBA} that runs in time $\Ostar(49^k)$.

\begin{restatable}{theorem}{main}
  \label{thm:main}
  \textsc{IMBA} can be solved in $\Ostar(49^k)$ time.
\end{restatable}

In other words, we identify a novel below-guarantee parameterization for \textsc{Induced Matching} smaller than that of \textsc{IMBTG} that yields an FPT algorithm, thereby positively answering our question ($\star$).
To our surprise, it turns out that an answer to our question arises from using the average of two upper bounds as an upper bound.

To prove \Cref{thm:main}, we give a branching algorithm in which the \emph{measure} $\frac{1}{2}(\MM(G) + \IS(G)) - \ell$ decreases by $\frac{1}{2}$ in every branching step.
As we will see, branching in a naïve way (which leads to an FPT algorithm for \textsc{IMBTG}) does not always decrease the measure.
To work around this issue, we develop branching rules based on the Gallai--Edmonds decomposition.
To establish the correctness of our algorithm, we reveal a structural property of graphs $G$ with $\frac{1}{2}(\MM(G) + \IS(G)) = \IM(G)$, which may be of independent interest.

\section{Preliminaries}

\subsection{Notation}

We denote the set $\{ 1, \dots, t \}$ of integers by $[t]$.
All graphs are simple and undirected.
For a graph $G$, let $V(G)$ and $E(G)$ denote the set of vertices and edges, respectively.
Let $v \in V(G)$ be a vertex in $G$ and let $X \subseteq V(G)$ be a vertex set.
We use $N(v)$ to denote the neighborhood of $v$ (the set of vertices adjacent to $v$) and $N(X) = \bigcup_{v \in V} N(v) \setminus X$ to denote the neighborhood of $X$.
Let $\deg(v) = |N(v)|$ denote the degree of $v$.
Let $G[X]$ denote the subgraph induced by $X$.
We use $G - X$ to denote $G[V(G) \setminus X]$, i.e., the graph obtained from $G$ by deleting $X$.
We use the shorthand $G - v$ for $G - \{ v \}$.
A triplet $(u, v, w)$ of pairwise adjacent vertices is a \emph{triangle}.
A vertex $v$ with $\deg(v) = 0$ is \emph{isolated}.
A vertex $v$ with $\deg(v) = 1$ is a \emph{pendant vertex}.
A pair $(u, v)$ of adjacent vertices with $\deg(u) = \deg(v) = 1$ is an \emph{isolated edge}.
A triangle is a \emph{pendant triangle} if $\deg(u) = \deg(w) = 2$.

In our algorithm, we apply \emph{reduction rules} and \emph{branching rules}.
Herein, a reduction rule (branching rule) is a polynomial-time procedure that given an instance $I$, returns an instance $I'$ (a set of instances $I_1, \dots, I_c$, respectively).
We say that a reduction rule (branching rule) is \emph{correct} if $I$ is equivalent to $I'$, i.e., $I$ is a yes-instance if and only if $I'$ is a yes-instance ($I$ is a yes-instance if and only if there is some $c' \in [c]$ such that $I_{c'}$ is a yes-instance, respectively).

\subsection{Parameterized complexity}

In parameterized complexity, each instance of a problem is equipped with a parameter, usually denoted by the symbol $k$.
We say that a parameterized problem is fixed-parameter tractable (FPT) if there is an algorithm that solves it in $f(k) \cdot |I|^{O(1)}$ time, where $f$ is a computable function only depending on $k$ and $|I|$ is the input size.
We denote the running time as $O^{\star}(f(k))$, where the polynomial factor in $|I|$ is suppressed in the $O^{\star}$ notation.
For more in-depth notions in parameterized complexity, we refer to the standard textbook~\cite{DBLP:books/sp/CyganFKLMPPS15}.

\subsection{Matching theory}

We use several results from matching theory (see e.g., the book of Lov{\'a}sz and Plummer \cite{LP86}).
In particular, K\H{o}nig's theorem and the Gallai--Edmonds structural theorem play important roles in the running time analysis for \Cref{thm:main}.
Recall that a matching $M$ in a graph $G$ is a set of pairwise non-incident edges.
If a vertex $v \in V(G)$ is incident to an edge in $M$, then $M$ \emph{covers} $v$.
If there is no edge in $M$ incident with $v$, then $M$ \emph{misses} $v$.
A matching covering every vertex of $G$ is said to be \emph{perfect}.
A matching covering all but one vertex of $G$ is said to be \emph{near-perfect}.
A graph $G$ is \emph{factor-critical} if  for every vertex $v \in V(G)$, $G - v$ has a perfect matching.

\begin{theorem}[K\H{o}nig's theorem] \label{thm:konig}
  For a bipartite graph $G$, $\MM(G) = \VC(G)$.
\end{theorem}
\begin{definition}
  \label{def:gallaiedmonds}
  The Gallai--Edmonds decomposition of a graph $G$ is a partition of $V(G)$ into $D(G)$, $A(G)$, and $C(G)$ where
  \begin{itemize}
    \item $D(G) = \{ v \in V(G) \mid \text{there exists a maximum matching missing $v$} \}$,
    \item $A(G) = N(D(G))$,
    \item $C(G) = V(G) \setminus (A(G) \cup D(G))$.
  \end{itemize}
\end{definition}

We simply write $D, A, C$ for $D(G), A(G), C(G)$, respectively, when $G$ is clear from context.

\begin{theorem}[The Gallai--Edmonds structure theorem (see e.g. \cite{LP86})]
  \label{thm:gallaiedmonds}
  The Gallai--Edmonds decomposition satisfies the following properties.
  \begin{itemize}
    \item
      Every connected component of $G[D]$ is factor-critical.
    \item
      $G[C]$ has a perfect matching.
    \item
      Let $G'$ be the bipartite graph obtained from $G[D \cup A]$ by contracting every connected component of $G[D]$ into one vertex and removing edges in $G[A]$.
      Then, $|N_{G'}(A')| > |A'|$ for every $A' \subseteq A$.
    \item
      A matching $M$ is a maximum matching if and only if the following hold:
      \begin{itemize}
        \item
          $M$ contains a near-perfect matching for every connected component of $G[D]$.
        \item
          Every vertex in $A$ is matched to a vertex in $D$.
        \item
          $M$ contains a perfect matching of $G[C]$.
      \end{itemize}
    \item
      The Gallai--Edmonds decomposition can be computed in polynomial time.
  \end{itemize}
\end{theorem}

\subsection{Cameron--Walker graphs}

A graph $G$ whose maximum matching size equals maximum induced matching size (that is, $\MM(G) = \IM(G)$) is referred to as a \emph{Cameron--Walker graph} in the literature, as Cameron and Walker~\cite{DBLP:journals/dm/CameronW05a} gave the structural characterization of these graphs.
Recall that a triangle $(u, v, w)$ with $\deg(u) = \deg(w) = 2$ is called a pendant triangle.
A triangle star is a graph obtained from a triangle by adding any number of pendant triangles to one of its vertices (see \Cref{fig:cw}). (When we say that we add a pendant triangle to a vertex $v$, it means that we add two vertices $u$ and $w$ and add edges such that $u, v, w$ form a triangle.)

\begin{theorem}[\cite{DBLP:journals/dm/CameronW05a}]
  \label{thm:cameronwalker}
  For a connected graph $G$, $\MM(G) = \IM(G)$ if and only if one of the following holds:
  \begin{itemize}
    \item
      $G$ is a star.
    \item
      $G$ is a triangle star.
    \item
      $G$ is obtained from a connected bipartite graph $G'$ with a bipartition $V(G) = U \cup W$ by adding at least one pendant vertex to each vertex of $U$ and adding any number of pendant triangles to each vertex of $W$.
  \end{itemize}
\end{theorem}

See \Cref{fig:cw} for an illustration of a Cameron--Walker graph.
The statement provided by Cameron and Walker~\cite{DBLP:journals/dm/CameronW05a} had a small mistake.
The statement of \Cref{thm:cameronwalker} follows a slight modification of Hibi et al.~\cite{HibiHKK15}.
Note that \Cref{thm:cameronwalker} yields a linear-time algorithm to recognize Cameron--Walker graphs.

\begin{figure}[t]
  \begin{center}
    \vspace{-3ex}
    \begin{tikzpicture}[scale=0.7]
      \node[vertex] (s) at (0, 0) {};
      \node[vertex] (a1) at (20:2) {};
      \node[vertex] (a2) at (70:2) {};
      \node[vertex] (b1) at (110:2) {};
      \node[vertex] (b2) at (160:2) {};
      \node[vertex] (c1) at (200:2) {};
      \node[vertex] (c2) at (250:2) {};
      \node[vertex] (d1) at (290:2) {};
      \node[vertex] (d2) at (340:2) {};

      \draw (s) -- (a1) -- (a2) -- (s) -- (b1) -- (b2) -- (s) -- (c1) -- (c2) -- (s) -- (d1) -- (d2) -- (s);

      \begin{scope}[shift={(7, -1)}]
      \draw[thick, dotted, rounded corners] (-1.3, -.5) rectangle (-.7, 2.5) {};
      \node[vertex] (a1) at (-1, 0) {};
      \node[vertex] (a2) at (-1, 1) {};
      \node[vertex] (a3) at (-1, 2) {};
      \begin{scope}[shift={(-1, 0)}]
        \node[vertex] (a11) at (200:1.3) {};
        \node[vertex] (a12) at (160:1.3) {};
      \end{scope}
      \begin{scope}[shift={(-1, 2)}]
        \node[vertex] (a31) at (110:1.3) {};
        \node[vertex] (a32) at (150:1.3) {};
        \node[vertex] (a33) at (190:1.3) {};
        \node[vertex] (a34) at (230:1.3) {};
      \end{scope}
      \draw (a1) -- (a11) -- (a12) -- (a1);
      \draw (a3) -- (a31) -- (a32) -- (a3) -- (a33) -- (a34) -- (a3);

      \draw[thick, dotted, rounded corners] (.7, -.8) rectangle (1.3, 2.8) {};
      \begin{scope}[shift={(1, -.2)}]
        \node[vertex] (b1) at (0, 0) {};
        \node[vertex] (b11) at (0:1.3) {};
        \node[vertex] (b12) at (330:1.3) {};
      \end{scope}
      \node[vertex] (b2) at (1, .6) {};
      \node[vertex] (b21) at (2.3, .6) {};
      \node[vertex] (b3) at (1, 1.4) {};
      \node[vertex] (b31) at (2.3, 1.4) {};
      \begin{scope}[shift={(1, 2.2)}]
        \node[vertex] (b4) at (0, 0) {};
        \node[vertex] (b41) at (10:1.3) {};
        \node[vertex] (b42) at (350:1.3) {};
        \node[vertex] (b43) at (30:1.3) {};
      \end{scope}
      \draw (b11) -- (b1) -- (b12);
      \draw (b2) -- (b21);
      \draw (b3) -- (b31);
      \draw (b4) -- (b41); \draw (b4) -- (b42); \draw (b4) -- (b43);

      \draw (a1) -- (b1); \draw (a1) -- (b2); \draw (a1) -- (b4);
      \draw (a2) -- (b1); \draw (a2) -- (b3);
      \draw (a3) -- (b3); \draw (a3) -- (b4);
      \end{scope}
    \end{tikzpicture}
  \end{center}
  \caption{A triangle star with four pendant triangles (left). A Cameron--Walker graph (right). The vertex sets $U$ and $W$ are marked by dotted lines.}
  \label{fig:cw}
\end{figure}
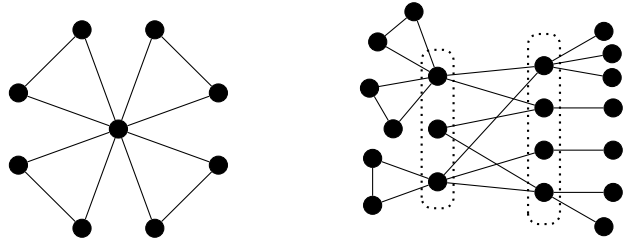

\subsection{FPT algorithm for \textsc{IMBTG}}
\label{ss:prem:algo}
In this subsection, we briefly discuss a simple branching algorithm that solves \textsc{IMBTG} in $\Ostar(9^k)$ time \cite{DBLP:journals/jda/MoserT09}.
Recall that the input of \textsc{IMBTG} is $(G, \ell)$ for a graph $G$ and $\ell \in \mathbb{N}$, and we search for an induced matching of size $\ell$ in $G$.
The parameterization of \textsc{IMBTG} is $k = \frac{1}{2} n - \ell$.
As we will see in \Cref{sec:algo}, the reduction rules and branching rules presented here form the basis for our algorithm for \textsc{IMBA}, which is parameterized by $\frac{1}{2}(\MM(G) + \IS(G)) - \ell$.

We begin with two reduction rules:
\begin{rrule}
  \label{rr:iso-vertex}
  Remove an isolated vertex.
\end{rrule}
\begin{rrule}
  \label{rr:iso-edge} \hspace{-.1em}
  Remove an isolated edge (including its endpoints) and decrease $\ell$ by one (if $\ell > 0$).
\end{rrule}
It is straightforward to prove the correctness of these reduction rules.
\begin{lemma}
  \label{lemma:rr:corr}
  \Cref{rr:iso-vertex,rr:iso-edge} are correct.
\end{lemma}

Since there is no vertex of degree at least two in any induced matching, we may branch into three instances as follows whenever there is a vertex of degree at least two:
\begin{brule}
  \label{brule:naive}
  Choose a vertex $v \in V(G)$ with at least two neighbors $u, w \in V(G)$.
  We branch into three instances: $(G - u, \ell)$, $(G - v, \ell)$, $(G - w, \ell)$.
\end{brule}
The algorithm terminates in one of the following two ways:
\begin{term}
  \label{tc:yes}
  Return yes if $\ell = 0$.
\end{term}
\begin{term}
  \label{tc:no}
  Return no if $\frac{1}{2}|V(G)| < \ell$.
\end{term}

In the simple $\Ostar(9^k)$-time algorithm for \textsc{IMBTG}, we repeat the following:
\begin{enumerate}
  \item Apply \Cref{rr:iso-vertex,rr:iso-edge} exhaustively.
  \item Check the termination conditions for \Cref{tc:yes,tc:no}.
  \item If there is a vertex of degree at least two, we branch according to \Cref{brule:naive}.
\end{enumerate}
If there is no vertex of degree at least two, then applying \Cref{rr:iso-vertex,rr:iso-edge} exhaustively deletes all vertices in the graph.
Thus, we always reach the condition for \Cref{tc:yes} ($\ell = 0$) or \Cref{tc:no} ($\ell > 0$).
For the running time analysis, note that our algorithm terminates if the value of $\frac{1}{2} |V(G)| - \ell$ is negative (\Cref{tc:no}).
This value decreases by $\frac{1}{2}$ each time we apply \Cref{brule:naive}:
$|V(G)|$ decreases by one and $\ell$ remains the same.
Moreover, this value does not increase by applying \Cref{rr:iso-vertex,rr:iso-edge}.
Since the value of $\frac{1}{2} |V(G)| - \ell$ equals $k$ for the input instance, we apply \Cref{brule:naive} at most $2k + 1$ times. (A more careful analysis yields an upper bound of $2k$.)
Hence, this algorithm runs in time $\Ostar(3^{2k + 1}) = \Ostar(9^k)$.

We remark that \Cref{brule:naive} can be strengthened as follows:
\begin{brule}
  \label{brule:naiveb}
  Choose two adjacent vertices $u, v \in V(G)$ with $N(\{ u, v \}) \ne \emptyset$.
  We branch into three instances: $(G - u, \ell)$, $(G - v, \ell)$, $(G - N(\{ u, v \}), \ell)$
\end{brule}

To see the correctness, observe that we have two cases:
If $G$ has an induced matching $M$ of size $\ell$ containing $uv$, then $M$ is an induced matching in $G - N(\{ u, v \})$ as well.
Otherwise, $G$ has no induced matching of size $\ell$ covering both $u$ and $v$, and thus $u$ or $v$ must be deleted.

Due to a space limitation, some proofs are in the appendix.

\section{Algorithm for IMBA}
\label{sec:algo}

In this section, we give an FPT algorithm \textsc{IMBA}:
Given a graph $G$ and an integer $\ell$, \textsc{IMBA} asks whether $G$ has an induced matching of size $\ell$ with the parameterization $\frac{1}{2}(\MM(G) + \IM(G)) - \ell$.
We start with an overview in \Cref{ss:algo:overview}.
We describe reduction rules in \Cref{ss:algo:rr} and branching rules in \Cref{ss:algo:br}.

\subsection{Overview}
\label{ss:algo:overview}
Our algorithm for \textsc{IMBA} is an intricate adaptation of the FPT algorithm given in \Cref{ss:prem:algo}.
In addition to \Cref{rr:iso-vertex,rr:iso-edge}, we use another reduction rule on pendant triangles (\Cref{rr:ts}).
We will also derive our branching rules from \Cref{brule:naive,brule:naiveb}.

To ensure that our algorithm runs in $\Ostar(49^k)$ time, we define the \emph{measure} of an instance $\mathcal{I} = (G, \ell)$ of \textsc{IMBA} as $\mu(\mathcal{I}) = \frac{1}{2}(\MM(G) + \IS(G)) - \ell$.
Note that the parameter $k$ for \textsc{IMBA} is the measure of the input instance.\footnote{As is often the case in parameterized complexity, we assume that $k$ is given along with input. We will not change the value of $k$ in this section.}
We design our algorithm such that (i) every branching rule generates at most seven instances whose measure is smaller by at least $\frac{1}{2}$ and (ii) there is no increase in the measure throughout.
This way, since we start with the measure at $k$, we end up with an instance whose measure is zero (or smaller) within $2k$ branching steps.\footnote{Our algorithm does not compute the measure. In fact, it is even computationally challenging to determine whether $\mu(\mathcal{I}) \le 0$ (this is equivalent to $\IS(G) \le 2\ell - \MM(G)$).}
As we show in \Cref{ss:algo:corr}, our algorithm correctly identifies yes-instances before the measure becomes zero or smaller.
So we terminate returning no after $2k$ branching steps:
\begin{term}
  \label{tc:2k}
  Return no if branching rules have been applied $2k$ times.
\end{term}
\Cref{tc:2k} ensures that the search tree has depth at most $2k$.
Our algorithm thus runs in $\Ostar(7^{2k}) = \Ostar(49^k)$ time.

Our algorithm repeats the following until one of the conditions for a termination is met.
\begin{enumerate}
  \item Apply \Cref{rr:iso-vertex,rr:iso-edge,rr:ts} exhaustively.
  \item Check the termination conditions for \Cref{tc:yes,tc:no,tc:2k}.
  \item If there is a vertex of degree at least two, we apply one of the branching rules in \Cref{ss:algo:br}.
\end{enumerate}
We remark that we check the condition of \Cref{tc:yes} before that of \Cref{tc:2k}.

As for how to branch, we want to branch in such a way that the measure drops by $\frac{1}{2}$.
To this end, branching according to \Cref{brule:naive} in the naïve way is seemingly not appropriate because the measure may not decrease.
This challenge is illustrated in \Cref{fig:branching}:
The maximum matching size and the maximum independent set both remain unchanged (thus so does the measure) after deleting the red vertex.
To find a set of vertices whose deletion guarantees a decrease in the measure by $\frac{1}{2}$, we will exploit the Gallai--Edmonds decomposition.

\begin{figure}[t]
  \begin{center}
    \begin{tikzpicture}[scale=1]
      \node[vertex, red] (s) at (0, 0) {};
      \node[vertex] (rb1) at (1, 0) {};
      \node[vertex] (rb2) at (2, 0) {};
      \node[vertex] (ru1) at (0.5, 1) {};
      \node[vertex] (ru2) at (1.5, 1) {};
      \node[vertex] (lu) at (-2, 1) {};
      \node[vertex] (lb) at (-2, 0) {};
      \node[vertex] (lm) at (-1.2, 0.5) {};
      \node[vertex] (su) at (-0.5, 1) {};

      \draw (ru1) -- (ru2);
      \draw (ru1) -- (s) -- (ru2);
      \draw (ru1) -- (rb1) -- (ru2);
      \draw (ru1) -- (rb2) -- (ru2);
      \draw (s) -- (su) -- (lm) -- (s);
      \draw (lu) -- (lm) -- (lb) -- (lu);
    \end{tikzpicture}
  \end{center}
  \caption{The maximum matching size and independent set size remain 4 even after deleting the red vertex.}
  \label{fig:branching}
\end{figure}
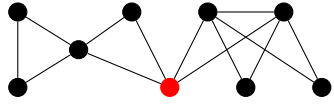

\subsection{Reduction Rules}
\label{ss:algo:rr}

Our algorithm employs \Cref{rr:iso-vertex,rr:iso-edge}.
We also introduce another reduction rule on pendant triangles.
Recall that a triplet $(u, v, w)$ is a pendant triangle if $u, v, w$ are pairwise adjacent and $\deg(u) = \deg(w) = 2$.

\begin{rrule}
  \label{rr:ts}
  If there is a pendant triangle $(u, v, w)$ with $\deg(u) = \deg(w) = 2$, then delete~$v$.
\end{rrule}
\begin{restatable}{lemma}{rrts}
  \label{lemma:rr:ts}
  \Cref{rr:ts} is correct.
\end{restatable}
\begin{proof}
  Let $G' = G - v$ be the graph obtained by deleting $v$.
  We show that $\IM(G) \ge \ell$ if and only if $\IM(G') \ge \ell$.
  First, observe that $\IM(G') \ge \ell$ implies $\IM(G) \ge \ell$, since any induced matching in $G'$ is also an induced matching in $G$.
  We then show that if $G$ has an induced matching $M$ of size $\ell$, then $\IM(G') \ge \ell$.
  If $M$ does not cover $v$, then $M$ is an induced matching in $G'$ as well, implying that $\IM(G') \ge \ell$.
  Suppose that $M$ covers $v$ and let $e$ denote the edge in $M$ incident to $v$.
  Then, the edge $uw$ is not part of $M$ since $u$ and $w$ are adjacent to $v$.
  Thus, $(M \setminus \{ e \}) \cup \{ uw \}$ is an induced matching of size $\ell$ in $G$ as well as $G'$.
\end{proof}

As discussed in \Cref{ss:algo:overview}, we need to ensure that these reduction rules do not increase the measure.
Recall that for an instance $\mathcal{I} = (G, \ell)$, its measure is $\mu(\mathcal{I}) = \frac{1}{2}(\MM(G) + \IS(G)) - \ell$.

\begin{restatable}{lemma}{rrmeasure}
  \label{lemma:rr:measure}
  Let $\mathcal{I} = (G, \ell)$ be an instance of \textsc{IMBA} and let $\mathcal{I}' = (G', \ell')$ be an instance obtained by applying Reduction Rule \ref{rr:iso-vertex}, \ref{rr:iso-edge}, or \ref{rr:ts}.
  Then, $\mu(\mathcal{I}') \le \mu(\mathcal{I})$.
\end{restatable}

\begin{proof}
  It is easy to see that when Reduction Rule \ref{rr:iso-vertex} or \ref{rr:ts} is applied, the measure does not increase as $\ell$ remains unchanged.
  There is no increase in the measure when applying \Cref{rr:iso-edge} either:
  We have $\MM(G') = \MM(G) - 1$ and $\IS(G') = \IS(G) - 1$ since every maximal matching contains an isolated edge and every maximal independent set contains exactly one endpoint of an isolated edge.
  Moreover, we have $\ell' = \ell - 1$.
  We thus have
  \begingroup
  \setlength{\belowdisplayskip}{0pt}
  \begin{align*}
    \mu(\mathcal{I}') = \frac{1}{2}(\MM(G') + \IS(G')) - \ell' = \frac{1}{2}(\MM(G) + \IS(G)) - \ell = \mu(\mathcal{I}).
  \end{align*}
  \endgroup
\end{proof}

\subsection{Branching Rules}
\label{ss:algo:br}

In this subsection, we describe our branching rules, which are based on the Gallai--Edmonds decomposition (\Cref{def:gallaiedmonds}).
Recall that the vertex set $V(G)$ is divided into three parts: $A$, $C$, and $D$. (Note that we have to recompute the Gallai--Edmonds decomposition after a vertex is deleted by our reduction rule or branching rule.)
In our algorithm, we apply the first reduction rule for which the condition is met in the following:
\begin{itemize}
  \item If $C \ne \emptyset$, then apply \Cref{br:ac}.
  \item If $A$ is not independent, then apply \Cref{br:aa}.
  \item If there is a connected component $S$ of $G[D]$ such that $G[S]$ is a triangle, then apply \Cref{br:triangle}. 
  \item If there is a connected component $S$ of $G[D]$ such that $G[S]$ is a triangle star, then apply \Cref{br:ts}. 
  \item If there is a connected component $S$ of $G[D]$ with $|S| \ge 5$, then apply \Cref{br:fc}. 
  \item If none of the above holds, then apply \Cref{br:naivea}.
\end{itemize}
For each branching rule, we prove that it is correct and that it decreases the measure by at least $\frac{1}{2}$.
We assume that reduction rules are exhaustively applied throughout the subsection.

In the first branching rule, we apply \Cref{brule:naive} on a vertex in $C$.
\begin{restatable}{brule}{brac}
  \label{br:ac}
  Choose a vertex $v \in C$ with at least two neighbors $u, w \in N(v)$.
  We branch into three instances $\mathcal{I}_i = (G_i, \ell)$ for $i \in [3]$, where $G_1 = G - v$, $G_2 = G - u$, and $G_3 = G - w$.
\end{restatable}
\begin{restatable}{lemma}{corrac}
  \label{lemma:corr:ac}
  \Cref{br:ac} is correct.
\end{restatable}
\begin{proof}
  The correctness follows from the correctness of \Cref{brule:naive}.
\end{proof}

\begin{restatable}{lemma}{aorc}
  \label{lemma:aorc}
  Let $v$ be a vertex in $A \cup C$ and let $G' = G - v$ be the graph obtained by deleting $v$.
  Then, $\MM(G') \le \MM(G) - 1$.
\end{restatable}
\begin{proof}
  By \Cref{thm:gallaiedmonds}, every vertex in $A \cup C$ is matched in every maximum matching.
  Since $G'$ misses one vertex from $A \cup C$, we have $\MM(G') \le \MM(G) - 1$.
\end{proof}

\begin{restatable}{lemma}{mac}
  \label{lemma:m:ac}
  In \Cref{br:ac}, $\mu(\mathcal{I}_i) \le \mu(\mathcal{I}) - \frac{1}{2}$ for all $i \in [3]$.
\end{restatable}
\begin{proof}
  By the definition of the Gallai--Edmonds decomposition, the neighbors of $v \in C$ are in $A \cup C$.
  It follows that $u, w \in A \cup C$.
  We thus have, by \Cref{lemma:aorc}, that $\MM(G_i) \le \MM(G) - 1$.
  We also have $\IS(G_i) \le \IS(G)$.
  Consequently, $\mu(\mathcal{I}_i) = \frac{1}{2} (\MM(G_i) + \IS(G_i)) - \ell \le \frac{1}{2} (\MM(G) - 1 + \IS(G)) - \ell = \mu(\mathcal{I}) - \frac{1}{2}$.
\end{proof}

We claim that if we cannot branch according to \Cref{br:ac} (that is, every vertex in $C$ has at most one neighbor), then $C = \emptyset$.
Assume for contradiction that there is a vertex $v \in C$.
Then, $v$ has a neighbor $u$ in $C$, since $G[C]$ has a perfect matching by \Cref{thm:gallaiedmonds}.
If neither $u$ nor $v$ has other neighbors, then \Cref{rr:iso-edge} applies.
Thus, we have $C = \emptyset$.

In the next rule, we branch on two adjacent vertices in $A$ adapting \Cref{brule:naiveb}.

\begin{restatable}{brule}{braa}
  \label{br:aa}
  Choose two adjacent vertices $u, v \in A$.
  We branch into three instances $\mathcal{I}_i = (G_i, \ell)$ for $i \in [3]$.
  In the first two, we delete $u$ or $v$, i.e., $G_1 = G - u$ and $G_2 = G - v$.
  In the third branch, we delete $N(\{ u, v \})$, i.e., $G_3 = G - N(\{ u, v \})$.
\end{restatable}

\begin{restatable}{lemma}{corraa}
  \label{lemma:corr:aa}
  \Cref{br:aa} is correct.
\end{restatable}

\begin{proof}
  By \Cref{thm:gallaiedmonds}, $N(\{ u, v \})$ has at least one vertex in $D$ (note that \Cref{brule:naiveb} requires $N(\{ u, v \}) \ne \emptyset$).
  The correctness of \Cref{br:aa} follows from the correctness of \Cref{brule:naiveb}.
\end{proof}

\begin{lemma}
  \label{lemma:m:aa}
  In \Cref{br:aa}, $\mu(\mathcal{I}_i) \le \mu(\mathcal{I}) - \frac{1}{2}$ for all $i \in [3]$.
\end{lemma}
\begin{proof}
  Since $G_1$ and $G_2$ each misses a vertex from $A$, we have $\MM(G_1) \le \MM(G) - 1$ and $\MM(G_2) \le \MM(G) - 1$ by \Cref{lemma:aorc}.
  We show that $\MM(G_3) \le \MM(G) - 1$ also holds.
  Any maximal matching $M$ of $G_3$ contains $uv$, since it is isolated.
  Note, however, that $u$ and $v$ are matched to vertices in $D$ in every maximum matching of $G$ by \Cref{thm:gallaiedmonds}.
  It follows that $M$ is not a maximum matching in $G$, implying that $\MM(G_3) \le \MM(G) - 1$.
  Since $G_i \, (i \in [3])$ arises from vertex deletions of $G$, we have $\IS(G_i) \le \IS(G)$.
  Consequently, $\mu(\mathcal{I}_i) = \frac{1}{2} (\MM(G_i) + \IS(G_i)) - \ell \le \frac{1}{2} (\MM(G) - 1 + \IS(G)) - \ell = \mu(G) - \frac{1}{2}$.
\end{proof}

Note that $A$ is an independent set when we cannot apply \Cref{br:aa}.
We will later describe how we branch on the vertices in $D$.
Before doing so, we show that deleting two vertices from a nontrivial (i.e., size at least two) connected component $S$ of $G[D]$ decreases the measure.
(In fact, $S$ is of size at least three because $G[S]$ is factor-critical by \Cref{thm:gallaiedmonds}.)

\begin{lemma}
  \label{lemma:d2}
  Let $S$ be a connected component of $G[D]$ with at least three vertices and let $u, v \in S$.
  Let $G' = G - u - v$ be the graph obtained from $G$ by deleting $u$ and $v$.
  Then, $\MM(G') \le \MM(G) - 1$.
\end{lemma}
\begin{proof}
  By \Cref{thm:gallaiedmonds}, every maximum matching of $G$ contains a near-perfect matching $M_S$ of $G[S]$.
  Since $G[S]$ is factor-critical, $S$ consists of an odd number of vertices, and we have $|M_S| = \frac{1}{2}(|S| - 1)$.
  Since $u, v \in S$ are deleted from $G'$, any matching in $G'[S]$ contains at most $\lfloor \frac{1}{2}(|S| - 2) \rfloor = \frac{1}{2}(|S| - 3) = |M_S| -1$ edges.
  It follows that $G'$ has no matching of size $\MM(G)$.
\end{proof}

For branching on the vertices in $D$, we start with the connected components of $G[D]$ which form a triangle.
We first prove a lemma concerning such triangles.

\begin{restatable}{lemma}{trihelper}
  \label{lemma:tri:helper}
  Let $S$ be a connected component of $G[D]$ such that $G[S]$ is a triangle.
  There exist two vertices $u, v \in S$ such that $u$ and $v$ have a neighbor $u'$ and $v'$ in $A$, respectively (possibly $u' = v'$).
\end{restatable}
\begin{proof}
  Since we have applied \Cref{rr:ts} exhaustively, there are two vertices $u, v \in S$ with degree at least three.
  Thus, each of $u$ and $v$ has at least one neighbor in $A$.
\end{proof}

\begin{restatable}{brule}{brtriangle}
  Let $S = \{ u, v, w \}$ be a connected component of $G[D]$ such that $G[S]$ is a triangle and $u$ and $v$ have a neighbor $u'$ and $v'$ in $A$, respectively (such vertices exist by \Cref{lemma:tri:helper}).
  Let $H_1 = G - u'$, $H_2 = G - u$, and $H_3 = G - v$.
  We generate seven instances $\mathcal{I}_i = (G_i, \ell)$ for $i \in [7]$, where $G_1 = H_1$, $G_2 = H_2 - v'$, $G_3 = H_2 - v$, $G_4 = H_2 - w$, $G_5 = H_3 - u'$, $G_6 = H_3 - u$, $G_7 = H_3 - w$.
  \label{br:triangle}
\end{restatable}
\begin{restatable}{lemma}{corrtriangle}
  \label{lemma:corr:triangle}
  \Cref{br:triangle} is correct.
\end{restatable}
\begin{proof}
  Observe that  $u$ is adjacent to $u'$ and $v$ in $G$.
  By the correctness of \Cref{brule:naive}, $\mathcal{I}$ is a yes-instance if and only if one of the three instances $\mathcal{J}_j = (H_j, \ell)$ is a yes-instance for $j \in [3]$.
  Since $v$ has two neighbors $v', w$ in $H_2$, $\mathcal{J}_2$ is a yes-instance if and only if $\mathcal{I}_i$ is a yes-instance for some $i \in [2, 4]$ by the correctness of \Cref{brule:naive}.
  Moreover, since $u$ has two neighbors $u', w$ in $H_3$, $\mathcal{J}_3$ is a yes-instance if and only if $\mathcal{I}_i$ is a yes-instance for some $i \in [5, 7]$ by the correctness of \Cref{brule:naive}.
  Hence, $\mathcal{I}$ is a yes-instance if and only if one of the seven instances $\mathcal{I}_i = (G_i, \ell)$ is a yes-instance for some $i \in [7]$.
  \end{proof}

We verify the drop in measure in \Cref{lemma:ts}.
We then look into the connected components of $G[D]$ which form triangle stars (with at least two pendant triangles).
We show a lemma analogous to \Cref{lemma:tri:helper}.

\begin{restatable}{lemma}{tshelper}
  \label{lemma:ts:helper}
  Let $S$ be a connected component of $G[D]$ such that $G[S]$ is a triangle star with at least two pendant triangles.
  There exist two nonadjacent vertices $u, v \in S$ such that $u$ and $v$ have a neighbor $u'$ and $v'$ in $A$, respectively (possibly $u' = v'$).
\end{restatable}
\begin{proof}
  Let $s$ be the center of $G[S]$, i.e., $s$ has more than two neighbors in $G[S]$.
  Then, by the assumption that \Cref{rr:ts} has been applied exhaustively, for every pendant triangle $(w_1, w_2, s)$ in $G[S]$, at least one of $w_1$ or $w_2$ has at least one neighbor in $A$.
  Thus, the lemma holds.
\end{proof}

\begin{restatable}{brule}{brts}
  \label{br:ts}
  Let $S$ be a connected component of $G[D]$ such that $G[S]$ is a triangle star with at least two pendant triangles.
  Let $u, v \in S$ be vertices as specified in \Cref{lemma:ts:helper}, and let $u',v' \in A$ be neighbors of $u, v$, respectively.
  Also, let $u'', v'' \in S$ be neighbors of $u, v$, respectively, which are not the center of $G[S]$.
  Let $H_1 = G - u'$, $H_2 = G - u$, and $H_3 = G - u''$.
  We generate seven instances $\mathcal{I}_i = (G_i, \ell)$ for $i \in [7]$, where $G_1 = H_1$, $G_2 = H_2 - v$, $G_3 = H_2 - v'$, $G_4 = H_2 - v''$, $G_5 = H_3 - v$, $G_6 = H_3 - v'$, $G_7 = H_3 - v''$.
\end{restatable}

Its correctness can be argued in the same way as we did for \Cref{br:triangle}.
\begin{lemma}
  \label{lemma:corr:ts}
  \Cref{br:ts} is correct.
\end{lemma}
\begin{lemma}
  \label{lemma:ts}
  In \Cref{br:triangle,br:ts}, $\mu(I_i) \le \mu(I) - \frac{1}{2}$ for all $i \in [7]$.
\end{lemma}
\begin{proof}
  Observe that we delete a vertex from $A$ or two vertices from $S$ in every branch.
  Thus, we obtain $\MM(G_i) \le \MM(G) - 1$ by \Cref{lemma:aorc,lemma:d2}.
  It follows that $\mu(\mathcal{I}_i) = \frac{1}{2}(\MM(G_i) + \IS(G_i)) - \ell = \frac{1}{2}((\MM(G) - 1) + \IS(G)) - \ell = \mu(\mathcal{I}) - \frac{1}{2}$.
\end{proof}

Finally, we look into nontrivial connected components of $G[D]$ which are not triangles or triangle stars.
We first prove two lemmas that help to develop a branching rule.
Note that each nontrivial component (unless it is a triangle) has at least five vertices.

\begin{restatable}{lemma}{fcpath}
  \label{lemma:fc-path}
  Let $H$ be a connected graph on at least four vertices that is not a star.
  Then, $H$ has a path on four vertices.
\end{restatable}

\begin{proof}
  Let $v$ be a vertex of maximum degree in $H$.
  We consider two cases:
  If $\deg(v) = |V(H)| - 1$, then by the assumption that $H$ is not a star, we have two adjacent vertices $w, x \in N(v)$.
  Since $|N(v)| = |V(H)| - 1 \ge 3$, there is a vertex $u \in N(v)$ distinct from $w$ and $x$.
  We thus have a path $(u, v, w, x)$.
  Otherwise, we have $\deg(v) < |V(H)| - 1$.
  We can assume that $\deg(v) \ge 2$ because $H$ is a connected graph on at least four vertices.
  Since $H$ is connected, there exist two adjacent vertices $w \in N(v)$ and $x \in V(H) \setminus (N(v) \cup \{ v \})$. 
  Since $\deg(v) \ge 2$, there is a vertex $u \in N(v)$ distinct from $w$.
  We thus have a path $(u, v, w, x)$.
\end{proof}

\begin{restatable}{lemma}{fcchoice}
  \label{lemma:fc-choice}
  Let $H$ be a factor-critical graph that is not a triangle star and let $v$ be an arbitrary vertex in $V(H)$.
  Then, $H - v$ has a vertex of degree at least two.
\end{restatable}
\begin{proof}
  Let $S_1, \dots, S_c$ be the connected components of $H - v$.
  Without loss of generality, assume that $|S_1| \ge \dots \ge |S_c|$.
  We show that $|S_1| \ge 3$.
  Since $H$ is factor-critical, $H - v$ has a perfect matching by definition.
  It follows that every connected component $S_i$ has at least two vertices.
  For the sake of contradiction, assume that $|S_i| = 2$ for every $i \in [c]$.
  We claim that for every connected component $S_i = \{ u_i, w_i \}$, $v$ is adjacent to both $u_i$ and $w_i$ in $H$. 
  If $v$ is not adjacent to $u_i$ (or $w_i$), then $G - w_i$ (or $G - v_i$, respectively) has no perfect matching because $u_i$ (or $w_i$, respectively) is isolated.
  Thus, $v$ is adjacent to every vertex in $H$.
  This, however, implies that $H$ is a triangle star, which contradicts our assumption.
  Thus, we have $|S_1| \ge 3$.
  Since $H[S_1]$ is connected, there is a vertex in $S_1$ of degree at least two in $H - v$.
\end{proof}

\begin{brule}
  \label{br:fc}
  Let $S$ be a nontrivial connected component of $G[D]$ with $|S| \ge 5$ which is not a triangle star.
  Choose a path $(u, v, w, x)$ on four vertices in $G[S]$ (such a path exists by \Cref{lemma:fc-path}).
  We have seven branches $\mathcal{I}_i = (G_i, \ell)$ as follows:
  \begin{enumerate}
    \item
      Choose a vertex $v'$ with at least two neighbors $v_1', v_2' \in S$ in $G - v$ (such a vertex exists by \Cref{lemma:fc-choice}).
      Let $G_1 = G - v - v'$, $G_2 = G - v - v_1'$, and $G_3 = G - v - v_2'$.
    \item
      Choose a vertex $w'$ with at least two neighbors $w_1', w_2' \in S$ in $G - w$ (such a vertex exists by \Cref{lemma:fc-choice}).
      Let $G_4 = G - w - w'$, $G_5 = G - w - w_1'$, and $G_6 = G - w - w_2'$.
    \item
      Let $G_7 = G - N(\{ v , w \})$.
  \end{enumerate}
\end{brule}
\par\vspace{-2.5ex}
\begin{lemma}
  \label{lemma:corr:fc}
  \Cref{br:fc} is correct.
\end{lemma}
\par\vspace{-1.5ex}
\begin{proof}
Observe that $\mathcal{I}$ is a yes-instance if and only if at least one of $(G - v, \ell)$, $G(G - w, \ell)$, and $G(G - N(\{ v, w \}), \ell)$ is a yes-instance by the correctness of \Cref{brule:naiveb}.
We branch further on the first two instances $(G - v, \ell)$ and $G(G - w, \ell)$ as in \Cref{brule:naive} to end up with the seven instances given above. 
\end{proof}
\par\vspace{-2.5ex}
\begin{lemma}
  \label{lemma:fc}
  In \Cref{br:fc}, $\mu(\mathcal{I}_i) \le \mu(\mathcal{I}) - \frac{1}{2}$ for all $i \in [7]$.
\end{lemma}
\par\vspace{-1.5ex}
\begin{proof}
  Note that every branch deletes at least two vertices from $S$ (in particular, $\{ u, x \} \subseteq N(\{ v, w\})$).
  We thus have $\MM(G_i) \le \MM(G) - 1$ by \Cref{lemma:d2}.
  It follows that $\mu(\mathcal{I}_i) = \frac{1}{2}(\MM(G_i) + \IS(G_i)) - \ell \le \frac{1}{2}(\MM(G) - 1 + \IS(G)) - \ell = \mu(\mathcal{I}) - \frac{1}{2}$.
\end{proof}

If the branching reduction rules given so far are not applicable, then $G$ is a bipartite graph whose partition is $A$ and $D$:
We have $C = \emptyset$, since otherwise we can apply \Cref{br:ac}.
We also have that $A$ is an independent set, since otherwise \Cref{br:aa} applies.
Moreover, every connected component $S$ of $G[D]$ is trivial:
Note that $S$ consists of an odd number of vertices because $G[S]$ is factor-critical.
If $|S| = 3$, then $S$ is a triangle and hence we can apply \Cref{br:triangle}.
If $|S| \ge 5$, then we can apply \Cref{br:ts} or \Cref{br:fc}.
Thus, every connected component of $G[D]$ is trivial, implying that $D$ is an independent set as well.
Now, \Cref{brule:naive} actually decreases the measure:

\begin{brule}
  \label{br:naivea}
  Choose a vertex $v \in V(G)$ with at least two neighbors $u, w \in V(G)$.
  We branch into $\mathcal{I}_i = (G_i, \ell)$ for $i \in [3]$, where $G_1 = G - v$, $G_2 = G - u$, and $G_3 = G - w$.
\end{brule}
\par\vspace{-1.5ex}
It is correct since \Cref{brule:naive} is correct.
\begin{lemma}
  \label{lemma:corr:naivea} 
  \Cref{br:naivea} is correct.
\end{lemma}
\begin{lemma}
  \label{lemma:m:naivea}
  In \Cref{br:naivea}, $\mu(\mathcal{I}_i) \le \mu(\mathcal{I}) - \frac{1}{2}$ for all $i \in [3]$.
\end{lemma}
\begin{proof}
  Recall that $\VC(H)$ denote the minimum vertex cover size of a graph $H$.
  Since $G$ and $G_i$ are both bipartite, we have $\VC(G) = \MM(G)$ and $\VC(G_i) = \MM(G_i)$ by K\H{o}nig's theorem (\Cref{thm:konig}).
  It follows that $\mu(\mathcal{I}) = \frac{1}{2}(\MM(G) + \IS(G)) - \ell = \frac{1}{2}(\VC(G) + \IS(G)) - \ell = \frac{1}{2}|V(G)| - \ell$.
  Analogously, we obtain $\mu(\mathcal{I}_i) = \frac{1}{2}|V(G_i)| - \ell$.
  Since $|V(G_i)| = |V(G)| - 1$, we have $\mu(\mathcal{I}) - \mu(\mathcal{I}_i) = \frac{1}{2}(|V(G)| - |V(G_i)|) = \frac{1}{2}$.
\end{proof}


\subsection{Correctness and Running Time Analysis}
\label{ss:algo:corr}

We prove \Cref{thm:main} in this subsection.
For the correctness of our algorithm (the outline is given in \Cref{ss:algo:overview}), we need to show that if the input $\mathcal{I}$ is a yes-instance, then our algorithm correctly determines that $\mathcal{I}$ is a yes-instance after branching $2k$ times (avoiding \Cref{tc:2k}).
We will show that if we end up with a yes-instance $\mathcal{I}' = (G', \ell')$ after $2k$ branching steps, then $\frac{1}{2}(\MM(G') + \IS(G')) = \IM(G')$ (equivalently, $\MM(G') = \IS(G') = \IM(G')$ since $\MM(G') \ge \IM(G')$ and $\IS(G') \ge \IM(G')$) holds. 
We examine the structure of graphs $G'$ with $\frac{1}{2}(\MM(G') + \IS(G')) = \IM(G')$ in \Cref{lemma:cw-adapt}.
We find that for a yes-instance $\mathcal{I}' = (G', \ell')$ with $\frac{1}{2}(\MM(G') + \IS(G')) = \IM(G')$, our algorithm terminates with $\ell = 0$ returning yes (\Cref{tc:yes}) after reduction rules are applied exhaustively (\Cref{lemma:algo:allrr}).

\begin{restatable}{lemma}{cwadapt}
  \label{lemma:cw-adapt}
  For a connected graph $G$, if $\frac{1}{2}(\MM(G) + \IS(G)) = \IM(G)$, then one of the following holds:
  \begin{itemize}
    \item
      $G$ is an isolated edge.
    \item
      $G$ is a triangle star.
    \item
      $G$ is obtained from a connected bipartite graph $G'$ with a bipartition $V(G') = U \cup W$ by adding exactly one pendant vertex to each vertex of $U$ and adding at least one pendant triangle to each vertex of $W$.
      In particular, $\MM(G) = \IS(G) = \IM(G) = \frac{1}{2}(|V(G)| - |W|)$.
  \end{itemize}
\end{restatable}
\begin{proof}
  Suppose that $G$ is a connected graph with $\frac{1}{2}(\MM(G) + \IS(G)) = \IM(G)$.
  We show that $G$ satisfies one of the above.
  Since $\MM(G) = \IS(G) = \IM(G)$, $G$ is a Cameron--Walker graph.
  We assume that $G$ has at least two vertices, since otherwise $\MM(G) = \IM(G) = 0$ and $\IS(G) = 1$.
  We examine each case of \Cref{thm:cameronwalker}.
  \begin{itemize}
    \item
      Suppose that $G$ is a star with $n - 1$ pendant vertices.
      It is easy to see that $\IS(G) = n - 1$ and $\IM(G) = 1$.
      Thus, we obtain $n = 2$.
    \item
      The second case of \Cref{thm:cameronwalker} states that $G$ is a triangle star.
      So we are done.
    \item
      Suppose that $G$ is obtained from a bipartite graph $G'$ with a bipartition $V(G') = U \cup W$ by adding at least one pendant vertex to each vertex of $U$ and adding any number of pendant triangles to each vertex of $W$.
      Let $n_u$ be the number of pendant vertices attached to $u$ for every $u \in U$ and let $n_w$ be the number of pendant triangles attached to $w$ for every vertex $w \in W$.
      We show that $n_u = 1$ for each $u \in U$ and $n_w \ge 1$ for each $w \in W$.

      First, we show that $\MM(G) = \IM(G) = |U| + \sum_{w \in W} n_w$.
      Let $M$ be a maximum (induced) matching of $G$.
      Since every vertex $u \in U$ has at least one pendant vertex adjacent to it, we can assume that $u$ is matched to one of its pendant neighbors in $M$.
      The deletion of $U$ and pendant vertices in $N(U)$ leaves $|W|$ triangle stars.
      The triangle star containing $w \in W$ has $n_w$ pendant triangles.
      Thus, $|M| = |U| + \sum_{w \in W} n_w$.

      We then show that $\IS(G) = \sum_{u \in U} n_u + \sum_{w \in W} \max(n_w, 1)$.
      We can assume that a maximum independent set $I$ contains all pendant vertices.
      Then, $I$ contains no vertex of $U$.
      After deleting all pendant vertices and their neighbors (that is, $U$), a triangle star with $n_w$ pendant triangles remains for each $w \in W$, which has an independent set of size $\max(n_w, 1)$.
      Thus, we have $\IS(G) = \sum_{u \in U} n_u + \sum_{w \in W} \max(n_w, 1)$.

      Since $n_u \ge 1$ and $\max(n_w, 1) \ge n_w$, we have $\IS(G) = \sum_{u \in U} n_u + \sum_{w \in W} \max(n_w, 1) \ge |U| + \sum_{w \in W} n_w = \IM(G)$.
      By the assumption that $\IS(G) = \IM(G)$, equality holds and hence $n_u = 1$ for each $u \in U$ and $\max(n_w, 1) = n_w$, that is, $n_w \ge 1$.
  \end{itemize}
  For the third case, note that since $|V(G)| = 2|U| + \sum_{w \in W} (2n_w + 1)$, we have $\MM(G) = \IM(G) = \IS(G) = |U| + \sum_{w \in W} n_w = \frac{1}{2}(|V(G)| - |W|)$.
\end{proof}

We remark that the converse of \Cref{lemma:cw-adapt} holds as well.
We show that an instance $(G, \ell)$ with $\frac{1}{2}(\MM(G) + \IS(G)) = \IM(G)$ can be solved by exhaustively applying reduction rules.

\begin{restatable}{lemma}{allrr}
  \label{lemma:algo:allrr}
  Let $\mathcal{I} = (G, \ell)$ be an instance of \textsc{IMBA} with $\frac{1}{2}(\MM(G) + \IS(G)) = \IM(G)$.
  If $\mathcal{I}$ is a yes-instance, then the instance $\mathcal{I'} = (G', \ell')$ obtained from $\mathcal{I}$ by exhaustively applying \Cref{rr:iso-vertex,rr:iso-edge,rr:ts} has $\ell' = 0$.
\end{restatable}
\begin{proof}
  Suppose that $\mathcal{I}$ is a yes-instance with $\frac{1}{2}(\MM(G) + \IS(G)) = \IM(G)$.
  Note that $\IM(G) \ge \ell$.
  We consider three cases of \Cref{lemma:cw-adapt}.
  \begin{itemize}
    \item
    Suppose that $G$ consists of an isolated edge.
    Then, $\ell \le 1$, and hence we have $\ell' = 0$ after an application of \Cref{rr:iso-edge}.
    \item
    Suppose that $G$ is a triangle star.
    Then, $\ell \le \frac{1}{2}(|V(G)| - 1)$, and hence we have $\ell' = 0$ after one application of \Cref{rr:ts} and $\frac{1}{2}(|V(G)| - 1)$ applications of \Cref{rr:iso-edge}.
    \item
    Suppose that $G$ arises from a connected bipartite graph with a bipartition $U \cup W$ as specified in \Cref{lemma:cw-adapt}.
    By \Cref{lemma:cw-adapt}, we have $\IM(G) = \frac{1}{2}(|V(G)| - |W|) \ge \ell$.
    Since every vertex $w \in W$ has at least one pendant triangle attached to it, \Cref{rr:ts} deletes every vertex in $W$.
    Note that after the deletion of $W$, we have a disjoint union of $\frac{1}{2}(|V(G)| - |W|)$ isolated edges.
    It follows that \Cref{rr:iso-edge} applies $\ell$ times, resulting in an instance $\mathcal{I}' = (G', \ell')$ with $\ell' = 0$.
  \end{itemize}
  This concludes the proof.
\end{proof}

Combining all our lemmas, we prove the following:

\main*

\begin{proof}
  We first show that our algorithm is correct.
  Let $\mathcal{S}$ be the set of instances corresponding to leaves in the search tree.
  For the correctness, we must show that if the input $\mathcal{I}$ is yes-instance, then there exists an instance $\mathcal{I}' \in \mathcal{S}$ which results in a termination returning yes (\Cref{tc:yes}).
  Note that each instance $\mathcal{I}' = (G', \ell') \in \mathcal{S}$ fulfills at least one of conditions for \Cref{tc:yes,tc:no,tc:2k}:
  Suppose that the condition for \Cref{tc:2k} not is met.
  Then, $G'$ does not have any vertex of degree at least two, since otherwise we can apply a branching rule in \Cref{ss:algo:br}.
  Thus, every vertex has degree at most one and \Cref{rr:iso-vertex,rr:iso-edge} deletes every vertex in the graph, resulting in \Cref{tc:yes} ($\ell' = 0$) or \Cref{tc:no} ($\ell' > 0$).

  If $\mathcal{I}$ is a no-instance, then every instance $\mathcal{I}' \in \mathcal{S}$ is a no-instance as well by the correctness of our reduction rules (\Cref{lemma:rr:corr,lemma:rr:ts}) and branching rules (\Cref{lemma:corr:ac,lemma:corr:aa,lemma:corr:triangle,lemma:corr:ts,lemma:corr:fc,lemma:corr:naivea}).
  Thus, our algorithm reaches \Cref{tc:no} or \ref{tc:2k}, returning no.
  
  Suppose that $\mathcal{I}$ is a yes-instance.
  Since our reduction rules and branching rules are correct, there exists an instance $\mathcal{I}' = (G', \ell') \in \mathcal{S}$ such that $\mathcal{I}'$ is a yes-instance.
  For the sake of contradiction, assume that our algorithm incorrectly concludes that $\mathcal{I}'$ is a no-instance, i.e., (i) $\ell' > 0$ (recall that we check for \Cref{tc:yes} first) and (ii) $\frac{1}{2} |V(G')| < \ell'$ or there have been $2k$ branching steps (conditions for \Cref{tc:no,tc:2k}, respectively).
  By the assumption that $\mathcal{I}'$ is a yes-instance, we have $\frac{1}{2} |V(G')| \ge \IM(G') \ge \ell'$.
  We thus may assume that branching rules have been applied $2k$ times.
  We have shown that every branching rule decreases the measure (\Cref{lemma:m:ac,lemma:m:aa,lemma:ts,lemma:fc,lemma:m:naivea}) by at least $\frac{1}{2}$ and that the measure does not increase by applying reduction rules (\Cref{lemma:rr:measure}).
  Since the measure is $k$ for the input instance $\mathcal{I}$, we have
  $\mu(\mathcal{I'}) \le k - 2k \cdot \frac{1}{2} = 0$.
  Moreover, we have
  \begin{align*}
    \mu(\mathcal{I}') = \frac{1}{2}(\MM(G') + \IS(G')) - \ell' \ge \frac{1}{2}(\MM(G') + \IS(G')) - \IM(G') \ge 0.
  \end{align*}
  Here, the first inequality follows because $\mathcal{I}'$ is a yes-instance, i.e., $\IM(G') \ge \ell'$, and the second inequality $\frac{1}{2}(\MM(G') + \IS(G')) - \IM(G') = \frac{1}{2}(\MM(G') - \IM(G')) + \frac{1}{2}(\IS(G') - \IM(G')) \ge 0$ holds for any graph $G'$.
  We thus have $\mu(\mathcal{I}') = 0$ and in particular $\frac{1}{2}(\MM(G') + \IS(G')) = \IM(G')$.
  Since \Cref{rr:iso-vertex,rr:iso-edge,rr:ts} are exhaustively applied, we have $\ell' = 0$ by \Cref{lemma:algo:allrr}, a contradiction.
  Thus, our algorithm correctly determines that $\mathcal{I}'$ is a yes-instance.
  
  For the running time, note that each node in the search tree has at most seven children.
  Moreover, the depth of the search tree is at most $2k$ by \Cref{tc:2k}.
  Thus, our algorithm runs in $\Ostar(7^{2k}) = \Ostar(49^k)$ time.
\end{proof}

\section{Hardness for IMBMM and IMBIS}
\label{sec:hard}

Here, we prove the hardness for \textsc{IMBMM} and \textsc{IMBIS}.
Recall that \textsc{IMBMM} is parameterized by $\MM(G) - \ell$ and \textsc{IMBIS} is parameterized by $\IS(G) - \ell$, where $\MM(G)$ is the maximum matching size of $G$, $\IS(G)$ is the maximum independent set size of $G$.
These negative results complement \Cref{thm:main} in answering our main question:
is there a parameterization smaller than $\frac{1}{2}n - \ell$ which admits an FPT algorithm?
Our negative results suggest that using $\MM(G)$ or $\IS(G)$ as an upper bound of $\ell$ fails.

We first show that \textsc{IMBMM} is W[2]-hard.
The hardness holds for 2-degenerate graphs.
This also complements the following two results:
One is an XP algorithm for \textsc{IMBMM} \cite{DBLP:journals/tcs/DuarteJPRS15} and the other is an FPT algorithm for \textsc{IMBMM} where the maximum degree is additionally included as part of the parameter \cite{DBLP:journals/tcs/DuarteJPRS15}.

\begin{restatable}{theorem}{imbmmhard}
  \textsc{IMBMM} is W[2]-hard even on 2-degenerate bipartite graphs.
\end{restatable}

\begin{proof}
  We reduce from \textsc{Dominating Set}, which is W[2]-hard:

  \problem
  {Dominating Set}
  {An undirected graph $G$ and $\ell \in \mathbb{N}$.}
  {Does $G$ have a dominating set $D$ (i.e., $N(D) = V(G) \setminus D$) of size at most $\ell$?}
  {$\ell$}

  Let $\mathcal{I} = (G, \ell)$ be an instance of \textsc{Dominating Set} such that $G$ is connected and $G$ has at least one cycle.
  Let $G'$ be the graph obtained from $G$ by subdividing every edge $vv'$ of $G$ (i.e., add a vertex $v''$, add two edges $vv''$ and $v'v''$, and delete the edge $vv'$).
  Let $U$ denote the set of vertices added by subdivisions.
  Note that every vertex in $U$ corresponds to an edge in $G$.
  It is straightforward to verify that $G$ is 2-degenerate:
  Let $X \subseteq V(G')$.
  If $X \cap U \ne \emptyset$, then $e \in X \cap U$ has degree at most two in $G[X]$.
  Otherwise, we have $X \subseteq U$ and thus $G[X]$ is a graph without any edge.
  Moreover, $G'$ is bipartite with a bipartition $V(G) \cup U$. 
  We show that $G$ has a dominating set of size $\ell$ if and only if $G'$ has an induced matching of size $\ell' = |V(G)| - \ell$.

  First, suppose that $G$ has a dominating set $D$ of size exactly $\ell$.
  By definition, every vertex $v \in V(G) \setminus D$ has at least one neighbor in $D$ in $G$.
  Let $v_D \in D$ be one of such neighbors.
  Now consider a matching $M'$ in $G'$ in which each vertex $v \in V(G) \setminus D$ is matched to the vertex (in $G'$) corresponding to the edge $v v_D$ (in $G$).
  We claim that $M'$ is an induced matching of size $\ell'$.
  Since $G'$ is bipartite, it suffices to verify that for two vertices $v, v' \in V(G) \setminus D$, there is no edge between $v_D$ and $v'$.
  By the choice of $v_D$, we have $N_{G'}(v_D) = \{ v, w \}$, where $w$ is some vertex in $D$, showing that $v_D$ and $v'$ are not adjacent to each other.
  We thus have shown that $M$ is an induced matching.
  Note that $M$ has size $|V(G) \setminus D| = |V(G)| - \ell$.

  Conversely, suppose that $G'$ has an induced matching $M'$ of size exactly $\ell'$.
  Since $G'$ is bipartite, $M'$ covers exactly $\ell'$ vertices of $V(G)$.
  Let $\overline{D} \subseteq V(G)$ be the set of vertices in $V(G)$ covered by $M'$.
  For every vertex $v \in \overline{D}$, let $v' \in U$ be the vertex such that $vv' \in M'$.
  Since $M'$ is an induced matching, $v'$ corresponds to an edge between $v$ and some vertex not in $\overline{D}$.
  This implies that $V(G) \setminus \overline{D}$ is a dominating set of size $|V(G)| - \ell' = \ell$ in $G$.
  
  To establish the W[2]-hardness of \textsc{IMBMM}, it remains to show that the parameter $k = \MM(G') - \ell'$ associated with the instance $(G', \ell')$ is upper-bounded by some function of $\ell$.
  In fact, we show that $k \le \ell$.
  To do so, we show that $\MM(G') = |V(G)|$, i.e., there is a matching in $G'$ that covers every vertex in $V(G)$ using Hall's theorem.
  To apply Hall's theorem, we must show that $|N(S)| \ge |S|$ holds in $G'$ for each $S \subseteq V(G)$.
  Note that $|N(S)|$ in $G'$ equals the number of edges incident to $S$ in $G$.
  For $S \subseteq V(G)$, let $S_1, \dots, S_c$ denote the connected component of $G[S]$.
  We claim that $|N(S_i)| \ge |S_i|$ for each $i \in [c]$.
  If $S_i = V(G)$, then by the assumption that $G$ has at least one cycle, we have $|N(S_i)| = |E(G)| \ge |V(G)| = |S_i|$.
  Assume that $S_i \ne V(G)$.
  For each $i \in [c]$, note that there are at least $|S_i| - 1$ edges in $G[S_i]$.
  Note also that there is at least edge in $G$ where one endpoint is in $S_i$ and the other is in $V(G) \setminus S_i$ (which is nonempty) by the assumption that $G$ is connected.
  Since $|N(S_i)|$ is the number of edges incident to $S_i$ in $G$, we have $|N(S_i)| \ge (|S_i| - 1) + 1 = |S_i|$ for each $i \in [c]$. 
  The sets $N(S_1), \dots, N(S_c)$ are pairwise disjoint, since otherwise for two distinct connected components $S_i$ and $S_j$, an edge $e \in N(S_i) \cap N(S_{j})$ would connect $S_i$ and $S_j$ in $G$.
  Thus, we have $|N(S)| \ge \sum_{i \in [c]} |N(S_i)| \ge \sum_{i \in [c]} |S_i| = |S|$.
  By Hall's theorem, $G'$ has a matching that covers every vertex in $V(G)$, i.e., $\MM(G') = |V(G)|$.
  It follows that $k = \MM(G') - \ell' = |V(G)| - (|V(G)| - \ell) = \ell$.
\end{proof}

We then show that \textsc{IMBIS} is NP-hard even if the parameter is zero.
Our hardness holds even if a maximum independent set is given as part of input.

\begin{restatable}{theorem}{imbishard}
  \textsc{IMBIS} is NP-hard for $k = 0$ even if an independent set of maximum size is given as part of input.
\end{restatable}

\begin{proof}
  We reduce from the NP-hard \textsc{Multicolored Independent Set} problem:
  Its input is an undirected graph $G$ and $\ell \in \mathbb{N}$.
  Additionally, we are given a partition $(C_1, \dots, C_\ell)$ of $V(G)$ into $\ell$ cliques.
  The task is to determine whether $G$ has an independent set of size $\ell$.
  Consider the graph $G'$ obtained by introducing $\ell$ vertices $v_1, \dots, v_\ell$ and adding edges such that $N(v_i) = C_i$ for every $i \in [\ell]$.
  It is not difficult to see that $G$ has an independent set of size $\ell$ if and only if $G'$ has an induced matching of size $\ell$.
  Note that the set $\{ v_1, \dots, v_\ell \}$ is an independent set of size $\ell$ in $G'$ and that this is of maximum cardinality since the vertex set can be partitioned into $\ell$ cliques.
  Thus, we have $k = \ell - \IM(G) = 0$.
  Note that our hardness holds even if $\{ v_1, \dots, v_\ell \}$ is given as part of input for \textsc{Induced Matching}.
\end{proof}

\section{Conclusion}

In this work, we discovered a new parameter $\frac{1}{2}(\MM(G) + \IS(G)) - \ell$ for which \textsc{Induced Matching} is fixed-parameter tractable.
This parameter is smaller than below trivial guarantee $\frac{1}{2} n - \ell$.
Our main result states that \textsc{Induced Matching} is solvable in $\Ostar(49^k)$ time for $k = \frac{1}{2}(\MM(G) + \IS(G)) - \ell$.
This stands in contrast to our negative results: the W[2]-hardness when parameterized by $\MM(G) - \ell$ and the NP-hardness for $\IS(G) - \ell = 0$

There remain several natural questions for future research.
First, is \textsc{Induced Matching} fixed-parameter tractable for an even smaller parameter?
We remark that our negative results in \Cref{sec:hard} indicates that the upper bound cannot be as tight as the maximum matching size or the maximum independent set size.
Another question is whether the base 49 in the running time of our algorithm can be lowered.
It would also be interesting to study whether \textsc{Induced Matching} has a polynomial kernel with respect to $k = \frac{1}{2}(\MM(G) + \IS(G)) - \ell$.

To the best of our knowledge, we are the first to propose the below (or above) guarantee parameterization, where the bound from which the parameter is derived is the average of two values.
We believe that this framework will be successful for other problems as well.

\bibliographystyle{plain}
\bibliography{literature}

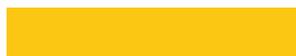
\end{document}